\documentclass[12pt,a4paper]{article}
\usepackage{amsfonts}
\usepackage{amsmath}

\newtheorem{theorem}{Theorem}

\newtheorem{definition}[theorem]{Definition}

\newtheorem{proposition}[theorem]{Proposition}
\newtheorem{remark}[theorem]{Remark}

\newenvironment{proof}[1][Proof]{\noindent\textbf{#1.} }{\ \rule{0.5em}{0.5em}}
\input{tcilatex}
\begin{document}

\title{Footprints and footprint analysis for atmospheric dispersion problems}
\author{Niklas Br\"{a}nnstr\"{o}m and Leif \AA\ Persson \\
Swedish Defence Research Agency, FOI\\
SE-901 82 Ume\aa \\
Sweden}
\maketitle
\date{}
\tableofcontents

\begin{abstract}
Footprint analysis, also known as the study of Influence areas, is a first
order method for solving inverse atmospheric dispersion problems. We revisit
the concept of footprints giving a rigorous definition of the concept
(denoted posterior footprints and posterior zero footprints) in terms of
spatio-temporal domains. The notion of footprints is then augmented the to
the forward dispersion problem by defining prior footprints and prior zero
footprints. We then study how posterior footprints and posterior zero
footprints can be combined to reveal more information about the source, and
how prior footprints and prior footprints can be combined to yield more
information about the measurements.
\end{abstract}

\section{Introduction}

In inverse atmospheric dispersion problems the task is to use sensor data
(e.g. concentration readings of a pollutant) to characterise the source of
the pollutant. As inverse problems are usually hard due the problem being
over-determined and having non-unique solutions many different methods to
solve the problem have been suggested. A common feature of these different
methods,\ however, is the tendency to specialise on certain
parameterisations of the problem. That is, the methods are usually tailored
to the problem at hand. This is of course a natural path to follow when
solving a particular problem, but in \cite{BrannstromPersson} an alternative
view was presented: therein a measure theoretic framework for studying
inverse atmospheric dispersion problems was developed. While not solving any
particular inverse dispersion problem, the framework allows for general
conclusions to be drawn about general inverse dispersion problems. (If a
particular problem is to be studied, the framework can be suitably
parameterised to cope with the situation). The framework itself relies on a
measure theoretic description of the dispersion problem (and its adjoint)
which is reviewed in Section \ref{Section:SettingofProblem}. In \cite%
{BrannstromPersson} this approach was employed to derive necessary and
sufficient conditions for when the inverse can be solved using the method of
least-squares. In the current paper we will use the same approach to study
another method for solving inverse atmospheric dispersion problems:
footprint analysis (which is also known as 'area of influence'). In essence,
footprint analysis is a rough method for solving the inverse problem where
the information contents of each sensor measurement is superpositioned to
gain a rough idea of the characteristics of the unknown source function.

With the risk of oversimplifying it seems footprint analysis can be divided
into two parts, where the first is to establish a relationship between the
source and the sensor and the second to use it to solve the inverse problem,
footprint analysis. In atmospheric dispersion models the analysis of how the
source influences the sensor is called footprints \cite{KRS2002}. The ideas
to study footprints (but without the nomenclature of today's literature)
were introduced in \cite{Pasquill1972}. Since then a whole body of
literature has emerged, with the survey article \cite{Schmid2002} being a
good starting point. Initially the focus lay on finding 2D footprints, i.e.
areas where the part of the source that influences the sensor most is
located. From the start the problem was set in terrains that were flat and
smooth but later the emphasis has been placed on studying the phenomena over
rougher terrain, in particular forest canopies, \cite{Schmid1994}, \cite%
{FWY1995}, \cite{KRS2002}, \cite{RAKSMV2000}. In \cite{CL2007} the footprint
analysis is generalised to 3D footprints (three spatial dimensions). To our
knowledge there are no spatio-temporal studies of the footprint despite many
dispersion problems being set both in space and time. In the present paper
we augment the notion of footprints to the spatio-temporal setting. The term
footprint analysis (or 'area of influence') is also used to refer to a first
order method for solving the inverse problem of estimating source
parameters, see e.g. \cite{Pudykiewicz1998} and \cite{Robertson2004}. Here
the idea is that using the adjoint formulation of the dispersion model one
can compute "back trajectories" starting from each sensor and evolving
backwards in time. The source, one then concludes, is located where all
these back trajectories intersect. The two uses of the word footprint stems
from exactly the same idea, namely to describe where the source is located,
but in the case of trying to solve the inverse problem it is always the
adjoint formulation of the dispersion model that is being used, and
secondly, if several measurements are available then this information is put
to use to solve the inverse problem (by intersecting the back trajectories).
In this paper we build on these ideas and use the framework of \cite%
{BrannstromPersson} to put these concepts on a rigorous footing. Indeed, the
footprints referred to in the footprint literature are closely related to
what we will term posterior footprints, and footprint analysis, used as a
first order method for solving the inverse problem, makes use of both
posterior footprints and posterior zero footprints. The nomenclature
regarding footprints will be explained in Section \ref{Section:Footprints},
but first we have to revise the setting of the underlying problem.

\section{An atmospheric dispersion problem, and its adjoint formulation\label%
{Section:SettingofProblem}}

The atmospheric dispersion problem that we are interested in can be
formulated in terms of a transition probability $p(t,x;s,y)$, where $%
(s,y),(t,x)\in T\times V$ where $T\subset \mathbb{R}$ is a time interval and 
$V\subset \mathbb{R}^{3}$ is a spatial domain. The transition probability
expresses the probability for a particle released at the time-space point $%
(s,y)$ to reside in the time-space point $(t,x)$ for $t\geq s$. We note that 
$p=0$ when $t<s$. The particles whose dispersion is governed by this
transition probability is assumed to originate from a source $S$. The source 
$S$ is assumed to be a positive measure on $T\times V$. In this way the
total mass\ $M$ released from the source is given by integrating the source
measure $S$ over its support%
\begin{equation}
M=\int_{T}\int_{V}dS(s,y).
\end{equation}%
The quantity that is usually desired as output from a dispersion model is
the concentration of the pollutant in a given space-time point. Since $S$
has its support on $T\times V$ and the transition probability describes the
dynamics of the released substance the concentration $c(t,x)$ is obtained by
weighing all released particles (released at some $(s,y)$ with $s<t$) with
the probability that they have been transported from $(s,y)$ to $(t,x)$%
\begin{equation}
c(t,x)=\int_{T}\int_{V}p(t,x;s,y)dS(s,y).
\end{equation}%
While $c(t,x)$ is the predicted concentration at the space time point $(t,x)$
the sensor may not have the resolution to make an ideal measurement from the
concentration field $c(t,x)$, indeed the sensor may perform some form of
averaging in both space and time to yield the sensor response $\overline{c}%
(t,x)$. We assume that the averaging process in the sensor can be described
by a probability measure $S^{\ast }$ (usually referred to as the
sensor-filter function) on $T\times V$, and hence we express the sensor
response as%
\begin{equation}
\overline{c}=\int_{T}\int_{V}c(t,x)dS^{\ast }(t,x).
\end{equation}%
Let us now use the definition of $c(t,x)$ to rewrite this expression in the
following way%
\begin{eqnarray}
\overline{c} &=&\int_{T}\int_{V}c(t,x)dS^{\ast }(t,x)  \notag \\
&=&\int_{T}\int_{V}\int_{T}\int_{V}p(t,x;s,y)dS(s,y)dS^{\ast }(t,x) \\
&=&\int_{T}\int_{V}\left( \int_{T}\int_{V}p(t,x;s,y)dS^{\ast }(t,x)\right)
dS(s,y).  \notag
\end{eqnarray}%
By defining the \emph{adjoint concentration field }$c^{\ast }(s,y)$ as%
\begin{equation}
c^{\ast }(s,y)=\int_{T}\int_{V}p(t,x;s,y)dS^{\ast }(t,x)  \label{def_adjoint}
\end{equation}%
we get%
\begin{equation}
\overline{c}=\int_{T}\int_{V}c^{\ast }(s,y)dS(s,y).
\end{equation}%
Hence we have two equivalent ways of calculating the sensor response%
\begin{equation}
\overline{c}=\int_{T}\int_{V}c(t,x)dS^{\ast
}(t,x)=\int_{T}\int_{V}\int_{V}c^{\ast }(s,y)dS(s,y)
\end{equation}%
which is the dual relationship between the forward and the adjoint
description of the dispersion problem. We note that equation (\ref%
{def_adjoint}) describing the adjoint concentration field is evolving
backwards in time: we may view the transition probability as moving adjoint
particles released by $S^{\ast }$ backwards in time and space. The main
advantage of using the adjoint representation in inverse dispersion
modelling is computational efficiency. This is a well-documented fact, see
for example \cite{Marchuk1986}. We also remark that the adjoint
concentration field $c^{\ast }$ is independent of the source function $S$,
and the concentration field $c$ is independent of the sensor-filter function 
$S^{\ast }$.

To better model real world sensors we assume that all sensors have a
threshold value $\widetilde{c}_{\lim }$ (the threshold value depends on the
specific sensor) below which any sensor response $\overline{c}$ will be put
to zero%
\begin{equation*}
\overline{c}=\left\{ 
\begin{array}{cc}
\overline{c}, & \text{if }\overline{c}\geq \widetilde{c}_{\lim } \\ 
0, & \text{otherwise}%
\end{array}%
\right. .
\end{equation*}

\section{Detectability and Non-detectability}

The dispersion problem predicts how a pollutant from a source spreads in the
atmosphere. From an abstract point of view this problem can be seen as a
problem of mapping of measures: the source $S$ can be viewed as a measure in
the spatio-temporal domain $T\times V$ that is being mapped via the
dispersion equations into a scalar function $c$ (the concentration), from
which we make measurements represented by a probability measure $S^{\ast }$,
defining the averaging of the concentration function $c$. From this level of
abstraction the adjoint version of the problem is very similar. In this case
the adjoint equations maps a probability measure $S^{\ast }$ on $T\times V$
representing a measurement in a sensor to a scalar function $c^{\ast }$
(adjoint 'concentration') from which we can make "adjoint measurements"
using a source measure $S$ acting on the adjoint 'concentration' $c^{\ast }$%
. (Depending on the scaling of the problem the adjoint 'concentration' $%
c^{\ast }$ may not be a proper concentration dimensionally.) In view of this
light asking questions about the sensor response in the forward problem or
asking questions about the source in the inverse problem are very similar.

Before heading into footprints we begin by studying the notion of
detectability, that is, when a source can be detected by sensor measurements.

\begin{definition}
A measurement $S^{\ast }$ is said to have sensitivity $k$ at $\left(
s,y\right) \in T\times V$ if $c^{\ast }\left( s,y\right) \geq k$.
\end{definition}

\begin{definition}
A measurement $S$ is said to detect the source $S$ at detection level $%
c_{\lim }$ if $\left\langle S,c^{\ast }\right\rangle \geq c_{\lim }$.
\end{definition}

To connect detection level to sensitivity we must assume a minimum mass of
the source. It follows from these definitions that

\begin{proposition}
If an instantaneous point source of mass $M$ at $\left( s,y\right) $ is
detected by measurement $S^{\ast }$, then $S^{\ast }$ has sensitivity $%
k=c_{\lim }/M$. To detect an instantaneous point source with mass at least $%
M_{\min }$, a sensitivity of $c_{\lim }/M_{\min }$ is required at the source
location.
\end{proposition}

A sensor detects on sensitivity level $k$ by weighting the concentration
field $c$ in a spatio-temporal neighbourhood of the sensor using the measure 
$S^{\ast }$. We state the some properties of this measurement in general
terms in the following theorem.

\begin{theorem}
\label{TheoremBackward}Assume that $S$ is a positive measure on a domain $%
D\subseteq \mathbb{R}^{n}$ and $c^{\ast }$ a nonnegative measurable function
on $D$, integrable with respect to $S$. Then for all $k\geq 0$%
\begin{equation}
kS\left\{ c^{\ast }\geq k\right\} +\int_{\left\{ c^{\ast }<k\right\}
}c^{\ast }dS\leq \int c^{\ast }dS  \label{lo_lim1}
\end{equation}%
and%
\begin{equation}
\int c^{\ast }dS\leq kS\left\{ c^{\ast }\leq k\right\} +\int_{\left\{
c^{\ast }>k\right\} }c^{\ast }dS  \label{hi_lim1}
\end{equation}%
Moreover, if there is equality in (\ref{lo_lim1}) then $S\left\{ c^{\ast
}>k\right\} =0$ (i.e., $c^{\ast }=k$, $S$--almost everywhere on $\left\{
c^{\ast }\geq k\right\} $). Finally, if there is equality in (\ref{hi_lim1}%
), then $S\left\{ c^{\ast }<k\right\} =0$ (i.e., $c^{\ast }=k$, $S$--almost
everywhere on $\left\{ c^{\ast }\leq k\right\} $).
\end{theorem}

\begin{proof}
We have 
\begin{equation}
\int c^{\ast }dS=\int_{\left\{ c^{\ast }\geq k\right\} }c^{\ast
}dS+\int_{\left\{ c^{\ast }<k\right\} }c^{\ast }dS\geq k\int_{\left\{
c^{\ast }\geq k\right\} }dS+\int_{\left\{ c^{\ast }<k\right\} }c^{\ast }dS
\end{equation}%
which proves equation (\ref{lo_lim1}). If there is equality in equation (\ref%
{lo_lim1}) we have 
\begin{equation*}
\int_{\left\{ c^{\ast }\geq k\right\} }c^{\ast }dS=k\int_{\left\{ c^{\ast
}\geq k\right\} }^{{}}dS
\end{equation*}%
which implies that $S\left\{ c^{\ast }>k\right\} =0$ (cf. \cite{Rudin1966},
Theorem 1.39 therein). The proof of equation (\ref{hi_lim1}) is similar,
with all inequalities reversed.
\end{proof}

\begin{definition}
A source $S$ is said to be $S^{\ast }$--detectable (with detection level $%
c_{\lim }$) if%
\begin{equation}
\int c^{\ast }dS\geq c_{\lim }  \label{def_detect}
\end{equation}%
and $S^{\ast }$--nondetectable if 
\begin{equation}
\int c^{\ast }dS<c_{\lim }  \label{def_nondetect}
\end{equation}
\end{definition}

Theorem \ref{TheoremBackward} gives necessary and sufficient conditions for $%
S$ to be $S^{\ast }$--detectable and $S^{\ast }$--nondetectable in terms of
masses on level sets $\left\{ c^{\ast }\geq k\right\} $ etc., which is
exploited in the following propositions.

\begin{proposition}
\label{Prop_necc_cond_det}(Necessary conditions for detection) Assume that $%
S $ is $S^{\ast }$--detectable and $k>0$. Then

\begin{enumerate}
\item 
\begin{equation}
kS\left\{ c^{\ast }\leq k\right\} +\int_{\left\{ c^{\ast }>k\right\}
}c^{\ast }dS\geq c_{\lim }  \label{necc_cond_det1}
\end{equation}

\item 
\begin{equation}
kS\left\{ c^{\ast }<k\right\} +\sup c^{\ast }S\left\{ c^{\ast }\geq
k\right\} \geq c_{\lim }  \label{necc_cond_det2}
\end{equation}

\item If there is a constant $\alpha \geq 0$ such that $S\left\{ c^{\ast
}<k\right\} \leq \alpha S\left\{ c^{\ast }\geq k\right\} $ then%
\begin{equation}
S\left\{ c^{\ast }\geq k\right\} \geq \frac{c_{\lim }}{k\alpha +\sup c^{\ast
}}  \label{necc_cond_det3}
\end{equation}

\item If $\sup c^{\ast }=\infty $, there is no positive lower bound on $%
S\left\{ c^{\ast }\geq k\right\} $, i.e., there are $S^{\ast }$--detectable
sources with arbitrarily small mass $S\left\{ c^{\ast }\geq k\right\} $.
\end{enumerate}
\end{proposition}

\begin{proof}
\begin{enumerate}
\item Equation (\ref{necc_cond_det1}) follows immediately from equations (%
\ref{hi_lim1}) and (\ref{def_detect}).

\item Since $c^{\ast }\leq \sup c^{\ast }$ we get from (\ref{necc_cond_det1}%
) that%
\begin{eqnarray}
c_{\lim } &\leq &kS\left\{ c^{\ast }\leq k\right\} +\int_{\left\{ c^{\ast
}>k\right\} }c^{\ast }dS  \notag \\
&\leq &kS\left\{ c^{\ast }\leq k\right\} +\sup c^{\ast }S\left\{ c^{\ast
}>k\right\}
\end{eqnarray}%
Moreover we have 
\begin{equation}
kS\left\{ c^{\ast }\leq k\right\} +\sup c^{\ast }S\left\{ c^{\ast
}>k\right\} \leq kS\left\{ c^{\ast }<k\right\} +\sup c^{\ast }S\left\{
c^{\ast }\geq k\right\}
\end{equation}%
(with equality if $S\left\{ c^{\ast }=k\right\} =0$, in particular, if $%
k>\sup c^{\ast }$), which proves (\ref{necc_cond_det2}).

\item We get from (\ref{necc_cond_det2}) and the additional condition $%
S\left\{ c^{\ast }<k\right\} \leq \alpha S\left\{ c^{\ast }\geq k\right\} $
that%
\begin{equation}
c_{\lim }\leq kS\left\{ c^{\ast }<k\right\} +\sup c^{\ast }S\left\{ c^{\ast
}\geq k\right\} \leq \left( k\alpha +\sup c^{\ast }\right) S\left\{ c^{\ast
}\geq k\right\}
\end{equation}%
which proves (\ref{necc_cond_det3}).

\item Take a sequence $t_{j},x_{j}$ of release times and locations such that 
$c^{\ast }\left( t_{j},x_{j}\right) \nearrow \infty $, and for each $j$ let $%
S_{j}$ be an instantaneous point source at $\left( t_{j},x_{j}\right) $ with
mass $M_{j}=c_{\lim }/c^{\ast }\left( t_{j},x_{j}\right) \searrow 0$. Then $%
\int c^{\ast }dS_{j}=c_{\lim }$, so each $S_{j}$ is $S^{\ast }$--detectable.
\end{enumerate}
\end{proof}

\begin{proposition}
\label{Prop_suff_cond_det}(Sufficient conditions for detection) Assume $k>0$
and at least one of the following conditions 1-3 is satisfied:

\begin{enumerate}
\item 
\begin{equation}
kS\left\{ c^{\ast }\geq k\right\} +\int_{\left\{ c^{\ast }<k\right\}
}c^{\ast }dS\geq c_{\lim }  \label{suff_cond_det1}
\end{equation}

\item 
\begin{equation}
kS\left\{ c^{\ast }\geq k\right\} +\inf c^{\ast }S\left\{ c^{\ast
}<k\right\} \geq c_{\lim }  \label{suff_cond_det2}
\end{equation}

\item There is a constant $\beta \geq 0$ such that%
\begin{eqnarray}
S\left\{ c^{\ast }<k\right\} &\geq &\beta S\left\{ c^{\ast }\geq k\right\} 
\text{ and}  \notag \\
S\left\{ c^{\ast }\geq k\right\} &\geq &\frac{c_{\lim }}{k+\beta \inf
c^{\ast }}  \label{suff_cond_det3}
\end{eqnarray}%
Then $S$ is $S^{\ast }$--detectable.
\end{enumerate}
\end{proposition}

\begin{proof}
\begin{enumerate}
\item Equation (\ref{def_detect}) follows immediately from equations (\ref%
{suff_cond_det1}) and (\ref{lo_lim1}).

\item Equation (\ref{suff_cond_det1}) follows from (\ref{suff_cond_det2})
since $\inf c^{\ast }S\left\{ c^{\ast }<k\right\} \leq \int_{\left\{ c^{\ast
}<k\right\} }c^{\ast }dS$.

\item Equation (\ref{suff_cond_det2}) follows from (\ref{suff_cond_det3})
since 
\begin{equation*}
kS\left\{ c^{\ast }\geq k\right\} +\inf c^{\ast }S\left\{ c^{\ast
}<k\right\} \geq \left( k+\beta \inf c^{\ast }\right) S\left\{ c^{\ast }\geq
k\right\} \geq c_{\lim }
\end{equation*}
\end{enumerate}
\end{proof}

\begin{proposition}
\label{Prop:Necc_cond_non_det}(Necessary conditions for nondetection) Assume
that $S$ is $S^{\ast }$--nondetectable and that $k>0$. Then

\begin{enumerate}
\item 
\begin{equation}
kS\left\{ c^{\ast }\geq k\right\} +\int_{\left\{ c^{\ast }<k\right\}
}c^{\ast }dS<c_{\lim }  \label{necc_cond_nondet1}
\end{equation}

\item 
\begin{equation}
kS\left\{ c^{\ast }\geq k\right\} +\inf c^{\ast }S\left\{ c^{\ast
}<k\right\} <c_{\lim }  \label{necc_cond_nondet2}
\end{equation}

\item If there is a constant $\gamma \geq 0$ such that $S\left\{ c^{\ast
}<k\right\} \geq \gamma S\left\{ c^{\ast }\geq k\right\} $ then%
\begin{equation}
S\left\{ c^{\ast }\geq k\right\} <\frac{c_{\lim }}{k+\gamma \inf c^{\ast }}
\label{necc_cond_nondet3}
\end{equation}
\end{enumerate}
\end{proposition}

\begin{proof}
Contrapositive of Proposition \ref{Prop_suff_cond_det}.
\end{proof}

\begin{proposition}
\label{Prop:Suff_cond_non_det}(Sufficient conditions for non--detection)
Assume $k>0$ and at least one of the following conditions 1-3 are satisfied:

\begin{enumerate}
\item 
\begin{equation}
kS\left\{ c^{\ast }\leq k\right\} +\int_{\left\{ c^{\ast }>k\right\}
}c^{\ast }dS<c_{\lim }  \label{suff_cond_nondet1}
\end{equation}

\item 
\begin{equation}
kS\left\{ c^{\ast }<k\right\} +\sup c^{\ast }S\left\{ c^{\ast }\geq
k\right\} <c_{\lim }  \label{suff_cond_nondet2}
\end{equation}

\item There is a constant $\varepsilon \geq 0$ such that%
\begin{eqnarray}
S\left\{ c^{\ast }<k\right\} &\leq &\varepsilon S\left\{ c^{\ast }\geq
k\right\} \text{ and}  \notag \\
S\left\{ c^{\ast }\geq k\right\} &<&\frac{c_{\lim }}{k\varepsilon +\sup
c^{\ast }}  \label{suff_cond_nondet3}
\end{eqnarray}%
Then $S$ is $S^{\ast }$--nondetectable.
\end{enumerate}
\end{proposition}

\begin{proof}
Contrapositive of Proposition \ref{Prop_necc_cond_det}.
\end{proof}

\section{Posterior and prior footprints, posterior and prior zero footprints 
\label{Section:Footprints}}

We want to define the notions of footprint and zero footprint. A footprint
is, loosely speaking, a subset $F$ of spacetime where the total source mass
is larger than a specified limit $M_{\lim }$, i.e., 
\begin{equation}
\int_{F}dS\left( t,x\right) \geq M_{\lim }  \label{FootprintMassCondition}
\end{equation}%
for all source measures $S$ in a given admissible class $\mathcal{S}$.
Likewise, a zero footprint is a subset $Z$ of spacetime where the total
source mass is smaller than a specified limit 
\begin{equation}
\int_{Z}dS\left( t,x\right) <M_{\lim }  \label{ZeroFootprintMassCondition}
\end{equation}%
for all $S\in \mathcal{S}$. To be of interest, the footprints and zero
footprints should be associated not only to a fixed set $\mathcal{S}$ of
admissible sources, but moreover restricted to subsets of $\mathcal{S}$
determined by conditions on measured values. Hence, given an $m$--tuple of
measurements $\left( S_{1}^{\ast },...,S_{m}^{\ast }\right) $ and
corresponding adjoint fields $c_{j}^{\ast }\left( s,y\right) =\iint_{T\times
V}p\left( s,y;t,x\right) dS_{j}^{\ast }\left( t,x\right) $ we consider
conditions on the form%
\begin{equation}
\left\langle S,c_{j}^{\ast }\right\rangle \geq \widetilde{c}_{\lim ,j}\text{
or }\left\langle S,c_{j}^{\ast }\right\rangle <\widetilde{c}_{\lim ,j}\text{
for }j=1,...m  \label{MeasuredValuesConditionNatural}
\end{equation}%
where $\widetilde{c}_{\lim ,j}>0$ are given limits (sensor thresholds). We
could work with these conditions in the form stated, but for the application
we have in mind (and for the sake of brevity) it is convenient to rewrite
these conditions in a form where the inequalities in both conditions (\ref%
{MeasuredValuesConditionNatural}) go in the same direction. We achieve this
by letting $\hat{c}_{\lim ,j}:=$ $\widetilde{c}_{\lim ,j}$\ if $\left\langle
S,c_{j}^{\ast }\right\rangle \geq \widetilde{c}_{\lim ,j}$, and $\check{c}%
_{\lim ,j}:=-$ $\widetilde{c}_{\lim ,j}$ if $\left\langle S,c_{j}^{\ast
}\right\rangle <\widetilde{c}_{\lim ,j}$, thus we have 
\begin{subequations}
\label{MeasuredValuesConditionInTransit}
\begin{equation}
\left\langle S,c_{j}^{\ast }\right\rangle -\hat{c}_{\lim ,j}>0
\label{PositiveLim}
\end{equation}%
for $\hat{c}_{\lim ,j}>0$, and 
\end{subequations}
\begin{equation}
-\left\langle S,c_{j}^{\ast }\right\rangle -\check{c}_{\lim ,j}>0
\label{NegativeLim}
\end{equation}%
for $\check{c}_{\lim ,j}<0$. By letting%
\begin{equation*}
c_{\lim ,j}=\left\{ 
\begin{array}{ccc}
\hat{c}_{\lim ,j} & \text{if} & \left\langle S,c_{j}^{\ast }\right\rangle
\geq \widetilde{c}_{\lim ,j} \\ 
\check{c}_{\lim ,j} & \text{if} & \left\langle S,c_{j}^{\ast }\right\rangle <%
\widetilde{c}_{\lim ,j}%
\end{array}%
\right. 
\end{equation*}%
we combine (\ref{PositiveLim}) and (\ref{NegativeLim}) into 
\begin{equation}
c_{\lim ,j}\geq 0\text{ or }c_{\lim ,j}<0\text{ and}\func{sign}\left(
c_{\lim ,j}\right) \left( \left\langle S,c_{j}^{\ast }\right\rangle
-\left\vert c_{\lim ,j}\right\vert \right) \geq 0\text{ for }j=1,...,m
\label{MeasuredValuesCondition}
\end{equation}%
where we define 
\begin{equation}
\func{sign}\left( c\right) =1_{\left\{ c\geq 0\right\} }-1_{\left\{
c<0\right\} }  \label{signDefinition}
\end{equation}%
We note that the limit $c_{\lim ,j}$ has the same physical interpretation as 
$\widetilde{c}_{\lim ,j}$, the value of $c_{\lim ,j}$ is the limit
(threshold) while the sign of $\widetilde{c}_{\lim ,j}$ tells whether the
limit is exceeded (+) or not (-). Hence we represent lower limits by
positive values of $c_{\lim ,j}$\ and upper limits by negative values of $%
c_{\lim ,j}$. We define a footprint set $F$ by requiring a logical
implication between the footprint mass condition, equation (\ref%
{FootprintMassCondition}), and the measurement condition, equation (\ref%
{MeasuredValuesCondition}). Likewise, we define a zero footprint set $Z$ by
requiring a logical implication between the zero footprint mass condition,
equation (\ref{ZeroFootprintMassCondition}), and equation (\ref%
{MeasuredValuesCondition}). If the mass condition is necessary for the
measurement condition, we say that we have a posterior footprint or
posterior zero footprint; if the mass condition is sufficient, we say that
we have a prior footprint or prior zero footprint. Hence, posterior
footprints and posterior zero footprints are used to deduce facts about the
released masses, given the measurements, whilst prior footprints are used to
deduce facts about the measurements, given facts about the released masses.

More precisely, we have

\begin{definition}
A subset $F\subset T\times V$ is said to be a posterior $\left( S^{\ast
},c_{\lim },\mathcal{S},M_{\lim }\right) $--footprint (or posterior
footprint when the parameters are understood) if $S\in \mathcal{S}$ and $%
\func{sign}\left( c_{\lim ,j}\right) \left( \left\langle S,c_{j}^{\ast
}\right\rangle -\left\vert c_{\lim ,j}\right\vert \right) \geq 0$ for $%
j=1,...,m$ implies that $S\left( F\right) \geq M_{\lim }$. $F\subset T\times
V$ is said to be a prior $\left( S^{\ast },c_{\lim },\mathcal{S},M_{\lim
}\right) $--footprint (or prior footprint when the parameters are
understood) if $S\in \mathcal{S}$ and $S\left( F\right) \geq M_{\lim }$
implies that $\func{sign}\left( c_{\lim ,j}\right) \left( \left\langle
S,c_{j}^{\ast }\right\rangle -\left\vert c_{\lim ,j}\right\vert \right) \geq
0$ for $j=1,...,m$.
\end{definition}

\begin{definition}
A subset $Z\subset T\times V$ is said to be a posterior $\left( S^{\ast
},c_{\lim },\mathcal{S},M_{\lim }\right) $--zero footprint (or posterior
zero footprint when the parameters are understood) if $S\in \mathcal{S}$ and 
$\func{sign}\left( c_{\lim ,j}\right) \left( \left\langle S,c_{j}^{\ast
}\right\rangle -\left\vert c_{\lim ,j}\right\vert \right) \geq 0$ for $%
j=1,...,m$ implies that $S\left( F\right) <M_{\lim }$. $Z\subset T\times V$
is said to be a prior $\left( S^{\ast },c_{\lim },\mathcal{S},M_{\lim
}\right) $--zero footprint (or prior zero footprint when the parameters are
understood) if $S\in \mathcal{S}$ and $S\left( F\right) <M_{\lim }$ implies
that $\func{sign}\left( c_{\lim ,j}\right) \left( \left\langle S,c_{j}^{\ast
}\right\rangle -\left\vert c_{\lim ,j}\right\vert \right) \geq 0$ for $%
j=1,...,m$.
\end{definition}

\begin{remark}
Note that the vector $c_{\lim }$ in the previous definition can hold both
positive and negative elements, thus we are handling measurements (positive
elements) and non-measurements (negative elements) simultaneously.
\end{remark}

To see some examples, consider the case of one measurement.

\begin{proposition}
\label{prop:FootPrintExample}If $m=1$, $c_{\lim }>0$, $k>0$, $\alpha \geq 0$%
, $\sup c^{\ast }<\infty $ and \ 
\begin{equation*}
\mathcal{S=}\left\{ S\in \mathcal{M}^{+}:S\left\{ c^{\ast }<k\right\} \leq
\alpha S\left\{ c^{\ast }\geq k\right\} \right\}
\end{equation*}%
and 
\begin{equation*}
M_{\lim }=\frac{c_{\lim }}{k\alpha +\sup c^{\ast }}
\end{equation*}%
then $\left\{ c^{\ast }\geq k\right\} $ is a posterior $\left( S^{\ast
},c_{\lim },\mathcal{S},M_{\lim }\right) $--footprint and a prior $\left(
S^{\ast },-c_{\lim },\mathcal{S},M_{\lim }\right) $--zero footprint.
\end{proposition}

\begin{proof}
Proposition \ref{Prop_necc_cond_det} and Proposition \ref%
{Prop:Suff_cond_non_det}.
\end{proof}

\begin{proposition}
\label{prop:ZeroFootprintExample}If $m=1$, $c_{\lim }<0$, $k>0$, $\beta \geq
0$ and 
\begin{equation*}
\mathcal{S=}\left\{ S\in \mathcal{M}^{+}:S\left\{ c^{\ast }<k\right\} \geq
\beta S\left\{ c^{\ast }\geq k\right\} \right\}
\end{equation*}%
and 
\begin{equation*}
M_{\lim }=\frac{-c_{\lim }}{k+\beta \inf c^{\ast }}
\end{equation*}%
then $\left\{ c^{\ast }\geq k\right\} $ is a prior $\left( S^{\ast
},-c_{\lim },\mathcal{S},M_{\lim }\right) $--footprint and a posterior $%
\left( S^{\ast },c_{\lim },\mathcal{S},M_{\lim }\right) $--zero footprint.
Note that if $\beta =0$ then $\mathcal{S=M}^{+}$.
\end{proposition}

\begin{proof}
Proposition \ref{Prop_suff_cond_det} and Proposition \ref%
{Prop:Necc_cond_non_det}.
\end{proof}

The fact that we obtained pairs of prior/posterior footprints/zero
footprints in the preceding propositions is not a coincidence. Indeed, we
have

\begin{proposition}
\label{prop:PriorPosteriorPairs}Assume that $c_{\lim }^{\prime }\neq c_{\lim
}$ and that $c_{\lim ,j}^{\prime }=\pm c_{\lim ,j}$. Then

\begin{enumerate}
\item If $A$ is a posterior $\left( S^{\ast },c_{\lim },\mathcal{S},M_{\lim
}\right) $--footprint then $A$ is a prior $\left( S^{\ast },c_{\lim
}^{\prime },\mathcal{S},M_{\lim }\right) $--zero footprint.

\item If $A$ is a prior $\left( S^{\ast },c_{\lim },\mathcal{S},M_{\lim
}\right) $--footprint then $A$ is a posterior $\left( S^{\ast },c_{\lim
}^{\prime },\mathcal{S},M_{\lim }\right) $--zero footprint.

\item If $A$ is a posterior $\left( S^{\ast },c_{\lim },\mathcal{S},M_{\lim
}\right) $--zero footprint then $A$ is a prior $\left( S^{\ast },c_{\lim
}^{\prime },\mathcal{S},M_{\lim }\right) $--footprint.

\item If $A$ is a prior $\left( S^{\ast },c_{\lim },\mathcal{S},M_{\lim
}\right) $--zero footprint then $A$ is a posterior $\left( S^{\ast },c_{\lim
}^{\prime },\mathcal{S},M_{\lim }\right) $--footprint.
\end{enumerate}

\begin{proof}
The condition $\limfunc{sign}\left( c_{\lim ,j}\right) \left( \left\langle
S,c_{j}^{\ast }\right\rangle -\left\vert c_{\lim ,j}\right\vert \right) \geq
0$ is not fulfilled for all $j$ if and only if it is violated for at least
one component, i.e., $\limfunc{sign}\left( c_{\lim ,j}^{\prime }\right)
\left( \left\langle S,c^{\ast }\right\rangle -\left\vert c_{\lim ,j}^{\prime
}\right\vert \right) \geq 0$ for some $c_{\lim }^{\prime }\neq c_{\lim }$
with $c_{\lim ,j}^{\prime }=\pm c_{\lim ,j}$, so the results follows by
contraposition.
\end{proof}
\end{proposition}

Let us now investigate how we can construct new footprints from old ones by
set theory operations. Some are obvious, collected in the following

\begin{proposition}
\label{prop:inclusions}
\end{proposition}

\begin{enumerate}
\item If $F$ is a posterior $\left( S^{\ast },c_{\lim },\mathcal{S},M_{\lim
}\right) $--footprint, $\mathcal{S}^{\prime }\subseteq \mathcal{S}$, $%
M_{\lim }^{\prime }\leq M_{\lim }$ and $F^{\prime }\supseteq F$, then $%
F^{\prime }$ is a $\left( S^{\ast },c_{\lim },\mathcal{S}^{\prime },M_{\lim
}^{\prime }\right) $--footprint.

\item If $F$ is a prior $\left( S^{\ast },c_{\lim },\mathcal{S},M_{\lim
}\right) $--footprint, $\mathcal{S}^{\prime }\subseteq \mathcal{S}$, $%
M_{\lim }^{\prime }\geq M_{\lim }$ and $F^{\prime }\subseteq F$, then $%
F^{\prime }$ is a prior $\left( S^{\ast },c_{\lim },\mathcal{S}^{\prime
},M_{\lim }^{\prime }\right) $--footprint.

\item If $Z$ is a posterior $\left( S^{\ast },c_{\lim },\mathcal{S},M_{\lim
}\right) $--zero footprint, $\mathcal{S}^{\prime }\subseteq \mathcal{S}$, $%
M_{\lim }^{\prime }\geq M_{\lim }$ and $F^{\prime }\subseteq F$, then $%
F^{\prime }$ is a posterior $\left( S^{\ast },c_{\lim },\mathcal{S}^{\prime
},M_{\lim }^{\prime }\right) $--zero footprint.

\item If $Z$ is a prior $\left( S^{\ast },c_{\lim },\mathcal{S},M_{\lim
}\right) $--zero footprint, $\mathcal{S}^{\prime }\subseteq \mathcal{S}$, $%
M_{\lim }^{\prime }\leq M_{\lim }$ and $F^{\prime }\supseteq F$, then $%
F^{\prime }$ is a prior $\left( S^{\ast },c_{\lim },\mathcal{S}^{\prime
},M_{\lim }^{\prime }\right) $--zero footprint.
\end{enumerate}

\begin{definition}
\begin{enumerate}
\item A set $F_{\min }\in T\times V$ is said to be a minimal posterior
footprint if there is no other posterior footprint $F$ (with the same
parameters) with $F\subset F_{\min }$.

\item A set $F_{\max }\in T\times V$ is said to be a maximal prior footprint
if there is no other prior footprint $F$ (with the same parameters) with $%
F\supset F_{\max }$.

\item A set $Z_{\max }\in T\times V$ is said to be a maximal posterior zero
footprint if there is no other posterior zero footprint $Z$ (with the same
parameters) with $Z\supset Z_{\max }$.

\item A set $Z_{\min }\subset T\times V$ is said to be a minimal prior zero
footprint if there is no other prior zero footprint $Z$ (with the same
parameters) with $Z\subset Z_{\min }$.
\end{enumerate}
\end{definition}

In the following proposition it is understood that all footprints are taken
with respect to the same parameters $\left( S^{\ast },c_{\lim },\mathcal{S}%
,M_{\lim }\right) $.

\begin{proposition}
\begin{enumerate}
\item For every nonempty posterior footprint $F$ there is a minimal
posterior footprint $F_{\min }$ with $F_{\min }\subseteq F$.

\item For every nonempty prior footprint $F$ there is a maximal prior
footprint $F_{\max }$ with $F\subseteq F_{\max }$.

\item For every nonempty posterior zero footprint $Z$ there is a maximal
posterior footprint $Z_{\max }$ with $Z\subseteq Z_{\max }$.

\item For every nonempty prior zero footprint $Z$ there is a minimal prior
zero footprint $Z_{\min }$ with $Z_{\min }\subseteq Z$
\end{enumerate}
\end{proposition}

\begin{proof}
Let $\mathcal{F}$ denote the class of all posterior $\left( S^{\ast
},c_{\lim },\mathcal{S},M_{\lim }\right) $--footprints. By Proposition \ref%
{prop:inclusions}, $\mathcal{F}$ is a partially ordered set with respect to
set inclusion. Consider a nonempty $F\in \mathcal{F}$ and a nest $\mathcal{N}
$ containing $F$, i.e., a subset $\mathcal{N\subset F}$ such that if $%
F_{1},F_{2}\in \mathcal{N}$, then either $F_{1}\subset F_{2}$ or $%
F_{2}\subset F_{1}$. By the Hausdorff Maximal Principle (see \cite{Kelley},
p. 32) $\mathcal{N}$ can be extended to a maximal nest in $\mathcal{F}$
(i.e., no other nest in $\mathcal{F}$ contains $\mathcal{N}$). Hence $%
F_{\min }=\cap _{F\in \mathcal{N}}F\in \mathcal{F}$ is a minimal element in $%
\mathcal{F}$ contained in $F$, i.e. there is no other $F\in \mathcal{F}$
contained in $F_{\min }$ as a proper subset. Likewise, $F_{\max }=\cup
_{F\in \mathcal{N}}F\in \mathcal{F}$ is a maximal element in $\mathcal{F}$
containing $F$. The proof for zero footprints is similar.
\end{proof}

These concepts are perhaps best illustrated for the case where the source
measures are point masses.

\begin{proposition}
Let $\mathcal{S}$ be the class of point masses on $T\times V$, $k>0$, $m=1$
and $c_{\lim }>0$. Then $Z=\left\{ c^{\ast }\geq k\right\} $ is a maximal
posterior $\left( S^{\ast },c_{\lim },\mathcal{S}\text{, }k/c_{\lim }\right) 
$--zero footprint, and any singleton set $F=\left\{ \left( s,y\right)
\right\} \subset \left\{ c^{\ast }\geq k\right\} $ is a minimal posterior $%
\left( S^{\ast },c_{\lim },\mathcal{S}\text{, }k/c_{\lim }\right) $%
--footprint.
\end{proposition}

\begin{proof}
$Z$ is a posterior $\left( S^{\ast },c_{\lim },\mathcal{S}\text{, }c_{\lim
}/k\right) $--zero footprint by Proposition \ref{prop:ZeroFootprintExample},
and $Z$ cannot be extended by some point $\left( s,y\right) $ with $c^{\ast
}\left( s,y\right) <k$, since then we would have $\left\langle k/c_{\lim
}\delta _{\left( s,y\right) },c^{\ast }\right\rangle =c_{\lim }c^{\ast
}\left( s,y\right) /k<c_{\lim }$, and hence there would be a point mass $S$
with slightly larger mass than $c_{\lim }/k$ which would still fulfill the
bound $\left\langle S,c^{\ast }\right\rangle <c_{\lim }$. $F$ is a posterior 
$\left( S^{\ast },c_{\lim },\mathcal{S}\text{, }c_{\lim }/k\right) $%
--footprint by Proposition \ref{prop:FootPrintExample}, and is clearly
minimal since it is a singleton set.
\end{proof}

\section{Footprint analysis, composite footprints}

Posterior and prior footprints as well as posterior and prior zero
footprints carries (spatio-temporal) information about the source measure
and the measurements. These footprints may have been calculated for subsets
of measurements separately which then gives rise to the question of whether
these pieces of information can be combined to give a more complete picture?
Let us begin by studying finite unions and finite intersections of
footprints:

\begin{proposition}
\label{prop:PosteriorFootprintUnion}If $F_{j}$ are posterior $\left(
S_{j}^{\ast },c_{\lim ,j},\mathcal{S},M_{\lim ,j}\right) $--footprints for $%
j=1,...,m$ and $F=\cup _{j}F_{j}$ then $F$ is a posterior $\left( S^{\ast
},c_{\lim },\mathcal{S},\max_{j}M_{\lim ,j}\right) $--footprint.
\end{proposition}

\begin{proof}
If $\func{sign}\left( c_{\lim ,j}\right) \left( \left\langle S,c_{j}^{\ast
}\right\rangle -\left\vert c_{\lim ,j}\right\vert \right) \geq 0$ then $%
S\left( F_{j}\right) \geq M_{\lim ,j}$, and $S\left( F\right) \geq S\left(
F_{j}\right) $ for all $j$, so $S\left( F\right) \geq \max_{j}M_{\lim ,j}$.
\end{proof}

\begin{proposition}
\label{prop:PosteriorZeroFootprintUnion}If $Z_{j}$ are posterior $\left(
S_{j}^{\ast },c_{\lim ,j},\mathcal{S},M_{\lim ,j}\right) $--zero footprints
for $j=1,...,m$ and $Z=\cup _{j}Z_{j}$ then $Z$ is a posterior $\left(
S^{\ast },c_{\lim },\mathcal{S},\sum_{j}M_{\lim ,j}\right) $--zero footprint.
\end{proposition}

\begin{proof}
If $\func{sign}\left( c_{\lim ,j}\right) \left( \left\langle S,c_{j}^{\ast
}\right\rangle -\left\vert c_{\lim ,j}\right\vert \right) \geq 0$ then $%
S\left( Z_{j}\right) <M_{\lim ,j}$, so $S\left( Z\right) \leq
\sum_{j}S\left( F_{j}\right) <\sum_{j}M_{\lim ,j}$.
\end{proof}

\begin{proposition}
\label{prop:PriorFootprintIntersection}If $F_{j}$ are prior $\left(
S_{j}^{\ast },c_{\lim ,j},\mathcal{S},M_{\lim ,j}\right) $--footprints for $%
j=1,...,m$ and $F=\cap _{j}F_{j}$ then $F$ is a prior $\left( S^{\ast
},c_{\lim },\mathcal{S},\max_{j}M_{\lim ,j}\right) $--footprint.
\end{proposition}

\begin{proof}
If $S\left( F\right) \geq \max_{j}M_{\lim ,j}$ then $S\left( F_{j}\right)
\geq M_{\lim ,j}$ for all $j$, and hence $\func{sign}\left( c_{\lim
,j}\right) \left( \left\langle S,c_{j}^{\ast }\right\rangle -\left\vert
c_{\lim ,j}\right\vert \right) \geq 0$ for all $j$.
\end{proof}

\begin{proposition}
\label{prop:PriorZeroFootprintUnion}If $Z_{j}$ are prior $\left( S_{j}^{\ast
},c_{\lim ,j},\mathcal{S},M_{\lim ,j}\right) $--zero footprints for $%
j=1,...,m$ and $Z=\cup _{j}Z_{j}$ then $Z$ is a prior $\left( S^{\ast
},c_{\lim },\mathcal{S},\min_{j}M_{\lim ,j}\right) $--zero footprint.
\end{proposition}

\begin{proof}
If $S\left( Z\right) <\min_{j}M_{\lim ,j}$ then $S\left( Z_{j}\right)
<M_{\lim ,j}$ for all $j$, and hence $\func{sign}\left( c_{\lim ,j}\right)
\left( \left\langle S,c_{j}^{\ast }\right\rangle -\left\vert c_{\lim
,j}\right\vert \right) \geq 0$ for all $j$.
\end{proof}

Concerning the set difference between a footprint and a zero footprint, we
have

\begin{proposition}
\label{prop:SetDifference}If $F$ is a $\left( S_{F}^{\ast },c_{F,\lim },%
\mathcal{S},M_{F,\lim }\right) $--posterior footprint and $Z$ is a $\left(
S_{Z}^{\ast },c_{Z,\lim },\mathcal{S},M_{Z,\lim }\right) $--posterior
footprint, then $F\setminus Z$ is a $\left( \left( S_{F}^{\ast },S_{Z}^{\ast
}\right) ,\left( c_{F,\lim },c_{Z,\lim }\right) ,\mathcal{S},M_{F,\lim
}-M_{Z,\lim }\right) $--posterior footprint.
\end{proposition}

\begin{proof}
If $\func{sign}\left( c_{F,\lim ,j}\right) \left( \left\langle
S,c_{F,j}^{\ast }\right\rangle -\left\vert c_{F,\lim ,j}\right\vert \right)
\geq 0$ and $\func{sign}\left( c_{Z,\lim ,j}\right) \left( \left\langle
S,c_{Z,j}^{\ast }\right\rangle -\left\vert c_{Z,\lim ,j}\right\vert \right)
\geq 0$ for all applicable $j$, then $S\left( F\right) \geq M_{F,\lim }$ and 
$S\left( Z\right) <M_{Z,\lim }$, so $S\left( F\setminus Z\right) =S\left(
F\right) -S\left( F\cap Z\right) \geq \left( S\left( F\right) -S\left(
Z\right) \right) >M_{F,\lim }-M_{Z,\lim }$.
\end{proof}

The following theorem shows the information that can be obtained from level
sets.

\begin{theorem}
\label{thm:main}Assume that $k_{j}>0$ and $0\leq \beta _{j}\leq \alpha _{j}$
for $j=1,...,m$, 
\begin{equation*}
\mathcal{S=}\dbigcap\limits_{j=1}^{m}\left\{ S\in \mathcal{M}^{+}:\beta
_{j}S\left\{ c_{j}^{\ast }\geq k_{j}\right\} \leq S\left\{ c_{j}^{\ast
}<k_{j}\right\} \leq \alpha _{j}S\left\{ c_{j}^{\ast }\geq k_{j}\right\}
\right\}
\end{equation*}%
and 
\begin{equation*}
M_{\lim ,j}=\max \left( \frac{c_{\lim ,j}}{k_{j}\alpha _{j}+\sup c_{j}^{\ast
}},\frac{-c_{\lim ,j}}{k_{j}+\beta _{j}\inf c_{j}^{\ast }}\right)
\end{equation*}%
Then%
\begin{equation*}
Z=\dbigcup\limits_{c_{\lim ,j}<0}\left\{ c_{j}^{\ast }\geq k_{j}\right\}
\end{equation*}%
is a posterior $\left( S^{\ast },c_{\lim },\mathcal{S},M_{Z,\lim }\right) $%
-zero footprint with%
\begin{equation*}
M_{Z,\lim }=\sum_{c_{\lim ,j}<0}M_{\lim ,j}=\sum_{c_{\lim ,j}<0}\frac{%
-c_{\lim ,j}}{k_{j}+\beta _{j}\inf c_{j}^{\ast }}
\end{equation*}%
and 
\begin{equation*}
F=\dbigcup\limits_{c_{\lim ,j}>0}\left\{ c_{j}^{\ast }\geq k_{j}\right\}
\end{equation*}%
is a posterior $\left( S^{\ast },c_{\lim },\mathcal{S},M_{F,\lim }\right) $%
-footprint with%
\begin{equation*}
M_{F,\lim }=\max_{c_{\lim ,j}>0}M_{\lim ,j}=\max_{c_{\lim ,j}>0}\frac{%
c_{\lim ,j}}{k_{j}\alpha _{j}+\sup c_{j}^{\ast }}
\end{equation*}%
Finally, $F\setminus Z$ is a posterior $\left( S^{\ast },c_{\lim },\mathcal{S%
},M_{\lim }\right) $-footprint with%
\begin{equation*}
M_{\lim }=M_{F,\lim }-M_{Z,\lim }=\max_{c_{\lim ,j}>0}\frac{c_{\lim ,j}}{%
k_{j}\alpha _{j}+\sup c_{j}^{\ast }}+\sum_{c_{\lim ,j}<0}\frac{c_{\lim ,j}}{%
k_{j}+\beta _{j}\inf c_{j}^{\ast }}
\end{equation*}
\end{theorem}

\begin{proof}
If $c_{\lim ,j}<0$, the set $\left\{ c_{j}^{\ast }\geq k_{j}\right\} $ is a
posterior $\left( S_{j}^{\ast },c_{\lim ,j}\mathcal{S},M_{\lim ,j}\right) $%
--zero footprint by Proposition \ref{prop:ZeroFootprintExample}, and hence $%
Z $ is a posterior $\left( S^{\ast },c_{\lim },\mathcal{S},M_{Z,\lim
}\right) $--zero footprint by Proposition \ref%
{prop:PosteriorZeroFootprintUnion}. Moreover, if $c_{\lim ,j}>0$ then the
set $\left\{ c_{j}^{\ast }\geq k_{j}\right\} $ is a posterior $\left(
S_{j}^{\ast },c_{\lim ,j}\mathcal{S},M_{\lim ,j}\right) $--footprint by
Proposition \ref{prop:FootPrintExample}, so $F$ is a $\left( S^{\ast
},c_{\lim },\mathcal{S},M_{F,\lim }\right) $--footprint by Proposition \ref%
{prop:PosteriorFootprintUnion}. Finally, $F\setminus Z$ is a $\left( S^{\ast
},c_{\lim },\mathcal{S},M_{\lim }\right) $--footprint by Proposition \ref%
{prop:SetDifference}.
\end{proof}

\section{Conclusion}

Using the measure theoretic framework introduced in \cite{BrannstromPersson}
we have provided rigorously defined the concept of footprints. Indeed, we
have defined posterior footprints, posterior zero footprints, prior
footprints and prior zero footprints. These footprints are all defined as
spatio-temporal domains. Based on the definitions we presented some basic
properties of the footprints, like the pairwise occurrence of
prior/posterior footprints/zero footprints, and maximal and minimal
footprints. We then studied how the information contents in single
footprints can be synthesised by taking finite unions and intersections of
footprints. The main result, Theorem \ref{thm:main}, shows how the posterior
zero footprint and posterior footprint are related to level lines of the
adjoint concentration fields $c_{j}^{\ast }$. Having adjoint concentration
fields $c_{j}^{\ast }$ is a common starting point of many methods of finding
solutions to inverse problems. Using Theorem \ref{thm:main} allows us to
immediately conclude in which part of the spatio-temporal domain we can
expect the source measure to have most of its (effective) weight ( the
posterior footprint), to have least of its (effective) weight (the posterior
zero footprint), and how to combine these footprints in an attempt to
further limit the spatio-temporal domain where most of the (effective)
weight of the source is located (the set difference of the posterior
footprint and the posterior zero footprint). We believe that this fast,
albeit rough, estimate of the source measure's spatio-temporal support will
be very useful in decision support systems that aid blue light forces when
handling CBRN events. Theorem \ref{thm:main} gives a first idea of what the
hazard area looks like, information that may be very desirable while the
more sophisticated inverse methods are busy calculating more refined hazard
areas and source estimates.

\bibliographystyle{alpha}
\bibliography{leiper}

\end{document}